\documentclass{amsart}
\usepackage{amsmath,amstext,amsthm,amsfonts,amssymb,latexsym}
\advance\textheight1.5\baselineskip

\theoremstyle{plain}
\newtheorem{theorem}{Theorem}
\newtheorem{corollary}[theorem]{Corollary}
\newtheorem{lemma}[theorem]{Lemma}
\newtheorem{proposition}[theorem]{Proposition}
\theoremstyle{remark}
\newtheorem*{remark*}{Remark}
\theoremstyle{definition}
\newtheorem*{example}{Example}

\newcommand\CC{{\mathbb C}}
\newcommand\RR{{\mathbb R}}
\newcommand\into{\int_\Omega}
\renewcommand\[{\begin{equation}}
\renewcommand\]{\end{equation}}
\newcommand\spr[1]{\langle#1\rangle}

\newcommand\hol{_{\text{hol}}}
\newcommand\ldvah{L^2_h}
\newcommand\ldvahh{L^2_{\text{hol},h}}
\newcommand\LL{{\mathcal L}}
\newcommand\HH{\mathcal H}
\newcommand\HHe{{\HH _\epsilon}}
\newcommand\intR{\int_\RR}
\renewcommand\Re{\operatorname{Re}}

\newcommand\tKe{\widetilde K_\epsilon}

\newcommand\Te{T^{(\epsilon)}}
\newcommand\tTe{\widetilde T^{(\epsilon)}}

\newcommand\tLe{\widetilde L_\epsilon}
\newcommand\tPe{\widetilde P_\epsilon}
\newcommand\jedna{\mathbf1}
\newcommand\oy{\overline y}
\newcommand\oz{\overline z}
\newcommand\PP{\mathcal P}
\newcommand\Peo{\PP_\epsilon^\Omega}
\newcommand\FFe{\mathcal F_\epsilon}
\newcommand\intC{\int_{\CC}}

\newcommand\LGe{\mathfrak L_\epsilon}
\newcommand\cR{\mathcal R}
\newcommand\ow{\overline w}
\newcommand\cT{\mathcal T}
\newcommand\dbar{\overline\partial{}}

\newcommand\La{L^{\alpha}}
\newcommand\LGea{\LGe^{(\alpha)}}
\newcommand\LLa{\LL^{(\alpha)}}
\newcommand\nua{\nu^{(\alpha)}}

\begin{document}

\title[Orthogonal polynomials and quasi-classical asymptotics]%
 {Orthogonal Polynomials, Laguerre Fock Space and Quasi-classical Asymptotics}
\author{Miroslav~Engli\v s}
\address{Mathematics Institute, Silesian University in Opava,
 Na~Rybn\'\i\v cku~1, 74601~Opava, Czech Republic {\rm and}
 Mathematics Institute, \v Zitn\' a 25, 11567~Prague~1, Czech Republic}
\email{englis{@}math.cas.cz}
\author{S.~Twareque Ali}
\address{Department of Mathematics and Statistics,
 Concordia University, Montr\'eal, Qu\'ebec, Canada~H3G~1M8}
\email{twareque.ali{@}concordia.ca}
\begin{abstract}
Continuing our earlier investigation of the Hermite case [J.~Math. Phys. 55
(2014), 042102], we~study an unorthodox variant of the Berezin-Toeplitz
quantization scheme associated with Laguerre polynomials. In~particular,
we~describe a ``Laguerre analogue'' of the classical Fock (Segal-Bargmann)
space and the relevant semi-classical asymptotics of its Toeplitz operators;
the~former actually turns out to coincide with the Hilbert space appearing
in the construction of the well-known Barut-Girardello coherent states.
Further extension to the case of Legendre polynomials is likewise discussed.
\end{abstract}
\maketitle

\section{Introduction}
One~of the very well studied methods of quantizing K\"ahler manifolds is the
Berezin-Toeplitz quantization \cite{BeQ,BMS}. In~the simplest case of a phase
space $\Omega$ admitting a global real-valued potential $\Psi$ (so~that the
K\"ahler form is given by $\omega=\partial\dbar\Psi$), one~considers the
$L^2$~space
\[ \ldvah = \{f\text{ measurable on }\Omega: \into |f|^2 e^{-\Psi/h}
 \,\omega^n <\infty \} \qquad(h>0),  \label{tTB}   \]
its subspace $\ldvahh$ of functions holomorphic on $\Omega$ (the~weighted
Bergman space), and the orthogonal projection $P_h:\ldvah\to\ldvahh$.
For~a~bounded measurable function $f$ on~$\Omega$, the~Toeplitz operator
$T_f$ on $\ldvahh$ with symbol $f$ is then defined~by
\[ T_f u = P_h(fu).   \label{tTC}   \]
This~is, in~fact, an~integral operator: more precisely, the space $\ldvahh$
turns out to be a reproducing kernel Hilbert space~\cite{Aro} possessing a
reproducing kernel $K_h(x,y)$, and (\ref{tTC}) can be rewritten~as
\[ T_f u (x) = \into u(y) f(y) K_h(x,y) \, e^{-\Psi(y)/h} \,\omega(y)^n .
 \label{tTD}  \]
When the manifold $\Omega$ is not simply connected, one has to assume that
the cohomology class of $\omega$ is integral, so~that there exists a Hermitian
line bundle $\LL$ with the canonical connection whose curvature form coincides
with~$\omega$; and the spaces $\ldvahh$ (and~$\ldvah$) get replaced by the
space of all holomorphic (or~all measurable, respectively) square-integrable
sections of~$\LL^{\otimes k}$, $k=\frac1h=1,2,3,\dots$. In~any case, under
reasonable technical assumptions on~$\Omega$ and~$\omega$, the~Toeplitz
operators satisfy
\[ T_f T_g \approx T_{fg}+h T_{C_1(f,g)} + h^2 T_{C_2(f,g)} + \dots \qquad
 \text{as } h\searrow0,   \label{tTE}   \]
with some bidifferential operators $C_j$ such that $C_1(f,g)-C_1(g,f)=\frac
i{2\pi}\{f,g\}$, implying in particular that the ``correspondence principle''
\[ T_f T_g - T_g T_f \approx \frac{ih}{2\pi} T_{\{f,g\}}
 \label{tTA}  \]
holds; here $\{\cdot,\cdot\}$ denotes the Poisson bracket.
Furthermore, the bidifferential operators $C_j$ can be expressed
in terms of covariant derivatives, with contractions of the curvature tensor
and its covariant derivatives as coefficients, thus encoding various geometric
properties of $(\Omega,\omega)$ in an intriguing~way. The~positive parameter
$h$ plays the role of the Planck constant.

The~purpose of the present paper, which is a sequel to \cite{AEH},
is~to highlight an operator calculus of a completely different flavour,
which nonetheless bears certain resemblance to~(\ref{tTA}) and~(\ref{tTD}),
and~arises in a quite unexpected setting --- namely, in~connection with
orthogonal polynomials.
To~be more specific, let $H_n(x)$ stand for the standard Hermite polynomials,
and, for $0<\epsilon<1$, set
\[ K_\epsilon(x,y) = \sum_{n=0}^\infty \epsilon^n \|H_n\|^{-2} H_n(x)
 \overline{H_n(y)}, \qquad x,y\in\RR.   \label{tTF}  \]
Here $\|H_n\|$ denotes the norm in $L^2(\RR,e^{-x^2}\,dx)$, where the $\{H_n\}$
form an orthogonal basis. Then $K_\epsilon$ is a positive-definite function,
and, hence, determines uniquely a Hilbert space $\HHe$ of functions on $\RR$
for which $K_\epsilon$ is the reproducing kernel~\cite{Aro}; this space has
been studied in \cite{karp} and was also encountered in \cite{AliKr} when
studying ``squeezed'' coherent states and their representations in terms of
Hermite polynomials of a complex variable. (The~definition of this kernel
may perhaps seem a bit artificial at first glance,
but~so must have seemed (\ref{tTB}) when it first came around in Berezin's
papers!) For~a (reasonable) function $f$ on~$\RR$, set
\[ T_f u(x) := \intR u(y) f(y) K_\epsilon(x,y) \, e^{-y^2}\,dy.  \label{tTG} \]
This certainly resembles the expression (\ref{tTD}) for Toeplitz operators,
however, note that this time there is no $L^2$ space around like~(\ref{tTB})
which would contain $\HHe$ as a closed subspace (in~fact, the set
$\{f(x)e^{-x^2/2}: f\in\HHe\}$ is a dense, rather than proper closed,
subset of~$L^2(\RR)$), so~there is no projection like $P_h$ around and
the original definition (\ref{tTC}) makes no sense.
In~particular, there is no reason \emph{a priori} even to expect (\ref{tTG})
to be defined, not to say bounded, on~some space (whereas with (\ref{tTC}) it
immediately follows that $\|T_f\|$ is not greater than the norm of the operator
of ``multiplication by~$f$'' on~$L^2$, hence $\|T_f\|\le\|f\|_\infty$).
It~may therefore come as a bit of a surprise that
(\ref{tTG}) actually yields, for $f\in L^\infty(\RR)$, a~bounded operator
on~$L^2(\RR)$, and, moreover, $T_f$~enjoys a nice asymptotic behaviour as
$\epsilon\nearrow1$, which we saw in \cite{AEH} to correspond, in~a~very
natural sense, to~the semiclassical limit $h\searrow0$ in the original
quantization setting.

Furthermore, it~turns out that the space $\HHe$ actually consists, up~to a
trivial equivalence, precisely of restrictions to $\RR$ of holomorphic
functions forming a very standard reproducing kernel space on the entire
complex plane~$\CC$. Namely, in~addition to being an orthogonal basis in
$L^2(\RR,e^{-|x|^2}dx)$, the Hermite polynomials also satisfy an orthogonality
relation over~$\CC$ \cite{karp,vanEijnd} :
\[ \intC H_n(z) \overline{H_m(z)} \; e^{-\frac{2\epsilon}{1+\epsilon} x^2
 -\frac{2\epsilon}{1-\epsilon} y^2} \; dx \, dy =
 \frac{\sqrt{1-\epsilon^2}}{2\epsilon} \, n!2^n\pi \epsilon^{-n} \delta_{mn},
 \qquad z=x+yi,   \label{tVI}  \]
It~follows that the multiplication operator
$$ M: f(z) \longmapsto \frac{\sqrt{2\epsilon}}{(1-\epsilon^2)^{1/4}\pi^{1/4}}
 e^{\frac{\epsilon^2}{1-\epsilon^2}z^2} f(z)  $$
maps the space $\HHe$ onto the space of holomorphic functions on~$\CC$ with
reproducing kernel
$$ F_\epsilon(z,w) := \frac{2\epsilon}{(1-\epsilon^2)\pi} K_\epsilon(z,w)
 = \frac{2\epsilon}{(1-\epsilon^2)\pi} e^{\frac{2\epsilon}{(1-\epsilon^2)}
 z\overline w} ,  $$
that~is, onto the standard Fock (Segal-Bargmann) space
$$ \FFe = L^2\hol(\CC, d\mu_\epsilon)   $$
of~all entire functions on $\CC$ square-integrable with respect to the
Gaussian measure
$$ d\mu_\epsilon(z) := e^{-2\epsilon|z|^2/(1-\epsilon)}\,dz,   $$
where $dz$ stands for the Lebesgue area measure on~$\CC$.
Now~$\FFe$ is precisely the space $\ldvahh$ as in~(\ref{tTB}) for $\Omega=\CC$
equipped with the standard (i.e.~Euclidean) K\"ahler structure.
Using the above correspondence between $\FFe$ and $L^2(\RR,e^{-|x|^2}dx)$,
one~can thus transfer the Toeplitz operators (\ref{tTD}) on $\FFe$ into
operators on $L^2(\RR,e^{-|x|^2}dx)$ and, via~another multiplication operator,
on~$L^2(\RR)$. The~latter turn out to belong to the standard Weyl calculus,
and~it was shown in \cite{AEH} that in this way one can actually recover,
from this seemingly totally unrelated \emph{Ansatz} involving Hermite
polynomials, the~whole Berezin-Toeplitz quantization (on~$\CC$) reviewed
in the beginnning.

In~the present paper, we~show that all the above, in~some sense, remains
in force also for the Laguerre polynomials $L_n$ in the place of~$H_n$.
In~particular, we~establish the existence of a certain analogue,
associated to the Laguerre polynomials, of~the Fock spaces~$\FFe$,
and study the semi-classical asymptotics of the Toeplitz operators there.
Surprisingly, this ``Laguerre Fock space'' turns out to coincide with the
space of entire functions discovered by Barut and Girardello \cite{BaG}
in the construction of coherent states that nowadays bear their~name.
(Similar spaces were also obtained in \cite{Ake} while working with
ensembles of non-Hermitian matrices and in \cite{jotha,karp,thanga}.)
The~associated Toeplitz operators and their asymptotics just mentioned,
however, up~to the authors' knowledge seem not to have previously appeared
in the literature: it~turns out that they again satisfy the correspondence
principle~\eqref{tTA}, but~with the Poisson bracket coming from the flat
metric on the punctured complex plane $\CC\setminus\{0\}$
(which is somewhat surprising).
We~also discuss the case of Legendre polynomials,
where things turn to work out somewhat differently.

The~necessary standard material on Laguerre polynomials is reviewed in
Section~2, and the associated reproducing kernel Hilbert spaces are
introduced there as well. The~Laguerre Fock space is discussed in Section~3,
and its Toeplitz operators in Section~5. A~result exhibiting the Laguerre
polynomials as a certain ``squeezed'' basis of the Laguerre Fock space is
discussed in Section~4. The~case of Legendre polynomials is analyzed
in Section~6, and some concluding remarks and speculations are collected
in the final Sections~7 and~8.

\section{Laguerre polynomials}
The~Laguerre polynomials $L_n(x)$, $n=0,1,2,\dots$, are defined~by the formula
$$ L_n(x) = \frac{e^x}{n!} \frac{d^n}{dx^n} x^n e^{-x} .   $$
They are orthonormal on the half-line $\RR_+=(0,+\infty)$ with respect to the
weight~$e^{-x}$; thus the functions
$$ l_n(x) := e^{-x/2} L_n(x), \qquad n=0,1,2,\dots,   $$
form an orthonormal basis of $L^2(\RR_+)$. They can also be obtained from
the generating function
\[ \sum_{n=0}^\infty L_n(x) z^n = \frac1{1-z} e^{\frac{xz}{z-1}},
 \qquad x\in\CC, |z|<1 .    \label{tLG}   \]
The~series
\[ L_\epsilon(x,y) = \sum_{n=0}^\infty \epsilon^n L_n(x) L_n(y),  \qquad
 \tLe(x,y) = \sum_{n=0}^\infty \epsilon^n l_n(x) l_n(y)
 = \frac{L_\epsilon(x,y)}{e^{(x+y)/2}}    \label{tRP}  \]
converge for all $x,y>0$, and
\[ L_\epsilon(x,y) = \frac1{1-\epsilon} e^{-\frac\epsilon{1-\epsilon}(x+y)}
 I_0\Big(\frac{2\sqrt{xy\epsilon}}{1-\epsilon}\Big),   \label{tRPx}  \]
where
$$ I_0(z) = \sum_{k=0}^\infty \Big(\frac{z^k}{k!2^k}\Big)^2   $$
is~the modified Bessel function; see~e.g.~\cite[\S6.2]{Askey}.
The~differential equation for Laguerre polynomials
$$ xL''_n(x) + (1-x)L'_n(x) + nL_n(x) =0   $$
is~equivalent to the equation
\[ Al_n=nl_n, \qquad Au(x):=-xu''(x)-u'(x)+\frac{x-2}4 u(x),
 \label{LAX}   \]
for the functions~$l_n(x)$.

Drawing inspiration from (\ref{tTD}), we~may define, for a function
(``symbol'') $f$ on~$\RR_+$, the corresponding ``Toeplitz operator''~$\tTe_f$,
$0<\epsilon<1$, on~$L^2(\RR_+)$~by
$$ \tTe_f u(x) := \int_0^\infty u(y) f(y) \tLe(x,y) \, dy .    $$
As~in \cite{AEH}, these turn out to be actually bounded operators, and possess
a kind of ``semi-classical'' asymptotic expansion as $\epsilon\nearrow1$.

\begin{theorem}
For $f\in L^\infty(\RR_+)$ the operator $\tTe_f$ is bounded on $L^2(\RR_+)$.
\end{theorem}

\begin{proof}
By~(\ref{tRP})
$$ \tTe_f u = \sum_n \epsilon^n \spr{fu,l_n}l_n .   $$
Thus, for any $0<\epsilon<1$,
$$ \|\tTe_f u\|^2 = \sum_n \epsilon^{2n} |\spr{fu,l_n}|^2
 \le \sum_n |\spr{fu,l_n}|^2 = \|fu\|^2 \le \|f\|_\infty^2 \|u\|^2,  $$
so $\|\tTe_f\|\le\|f\|_\infty$.
\end{proof}

\begin{theorem}
We~have
$$ \tTe_f = \epsilon^A M_f,  $$
where $A$ is as in~$(\ref{LAX})$, $\epsilon^A$~is understood in the sense
of the spectral theorem, and $M_f$ stands for the operator of
``multiplication by~$f$''. Consequently, as~$\epsilon\nearrow1$,
\[ \tTe_f u \approx \sum_{k=0}^\infty \frac{(\log\epsilon)^k}{k!} A^k(fu).
 \label{tRJ}  \]
\end{theorem}

\begin{proof}
We~have
\begin{align*}
\tTe_f u &= \int_0^\infty u(y) f(y) \sum_n \epsilon^n
  \overline{l_n(y)} l_n \, dy  \\
&= \sum_n \epsilon^n \spr{uf,l_n} l_n  \\
&= \sum_n \spr{uf,l_n} \epsilon^A l_n  \\
&= \epsilon^A(uf) = \sum_k \frac{(\log\epsilon)^k}{k!} \,A^k(uf) .
\end{align*}
\end{proof}

Of~course, using the familiar series
$$ \log\epsilon = -\sum_{j=1}^\infty \frac{(1-\epsilon)^j}j  $$
one could easily pass in (\ref{tRJ}) from powers of $\log\epsilon$ to powers
of $(1-\epsilon)$.

The~beginning of the asymptotic expansion (\ref{tRJ}) reads
$\tTe_fu=fu+(1-\epsilon)A(fu)+O((1-\epsilon)^2)$, or
\[ \tTe_f = M_f + (1-\epsilon)AM_f + O((1-\epsilon)^2) .  \label{tRK}  \]
Using the similar formulas for $g$ and $fg$ and subtracting, we~arrive~at
\[ \tTe_f \tTe_g - \tTe_{fg} = (1-\epsilon) M_fAM_g + O((1-\epsilon)^2) ,
 \label{tRL}  \]
and, upon a routine computation,
\[ \widetilde T_f \widetilde T_g - \widetilde T_g\widetilde T_f = (1-\epsilon)
 \big[ (x-1)(g(xf')'-f(xg')')I +2x(fg'-gf') D \big] + O((1-\epsilon)^2),
 \label{tRM}   \]
where we introduced the notation
$$ Du(x) := \frac{du(x)}{dx}   $$
for the differentiation operator on~$\RR$. Comparing these formulas with
(\ref{tTA}) and~(\ref{tTE}) --- the~role of the Planck constant being now
played by the quantity $1-\epsilon$ --- we~see that, first of~all, the~role
of the Poisson bracket is now played by the (second-order) expression
$g(xf')'-f(xg')'$; and, secondly, that in addition to the ``Toeplitz''
operators~$\tTe$, the~differentiation operator $D$ appears too.

As~with Hermite polynomials, one~also again has Hilbert spaces for which
$L_\epsilon$ and~$\tLe$ are the reproducing kernels:
\begin{align*}
 \LL_\epsilon &:= \{f=\sum_n f_n L_n: \sum_n \epsilon^{-n}|f_n|^2<\infty\} ,\\
 \widetilde\LL_\epsilon &:=
 \{f=\sum_n f_n l_n: \sum_n \epsilon^{-n}|f_n|^2<\infty\}.
\end{align*}
Here $\widetilde\LL_\epsilon$ is a space of functions on~$\RR_+$,
dense in $L^2(\RR_+)$. It~turns out that just as for the Hermite polynomials
in~\cite{AEH}, $\LL_\epsilon$ again extends to a space of holomorphic
functions on all of~$\CC$.

\begin{theorem}  \label{thmLag}
Each $f\in\LL_\epsilon$ extends to an entire function on~$\CC$,
and $\LL_\epsilon$ is the space of $($the~restrictions to~$\RR_+\!$~of\,$)$
holomorphic functions on $\CC$ with reproducing kernel
$$ L_\epsilon(x,y) = \sum_{n=0}^\infty \epsilon^n L_n(x) \overline{L_n(y)}
 = \frac1{1-\epsilon} e^{-\frac\epsilon{1-\epsilon}(x+\oy)}
 I_0\Big(\frac{2\sqrt{\epsilon x\oy}}{1-\epsilon}\Big), \qquad x,y\in\CC.  $$
\end{theorem}

\begin{proof}
By~(\ref{tLG}) and Cauchy estimates, we~have for each $0<r<1$ and $x\in\CC$
\[ r^n|L_n(x)| \le \sup_{|z|=r} \Big| \frac1{1-z} e^{\frac{xz}{z-1}} \Big|
 \le \frac1{1-r} e^{\frac{|x|r}{1-r}}.    \label{tLE}  \]
Thus
\begin{align*}
\sum_n |f_n L_n(x)| &\le \frac{e^{|x|r/(1-r)}}{1-r} \sum_n |f_n r^{-n}|  \\
&\le \frac{e^{|x|r/(1-r)}}{1-r} \;\|f\|_{\LL_\epsilon}
 \Big(\sum_n \epsilon^n r^{-2n} \Big)^{1/2}   \\
&\le \frac{e^{|x|r/(1-r)}}{1-r} \;\|f\|_{\LL_\epsilon}
 \frac r{\sqrt{r^2-\epsilon}}   \end{align*}
whenever $r\in(\sqrt\epsilon,1)$. Thus the series
$$ f(x) = \sum_n f_n L_n(x)   $$
converges for any $x\in\CC$, and uniformly on compact subsets.
The~rest follows as in the proof of Theorem~2 in~\cite{AEH}.   \end{proof}

\section{The Laguerre Fock space}
It~turns out that the Laguerre polynomials also satisfy an orthogonality
relation over the complex plane, similarly to (\ref{tVI}) for the Hermite
polynomials.

Recall that the modified Bessel function of the third kind $K_0$ is defined~by
\[ K_0(t) = \int_1^\infty \frac{e^{-tx}} {\sqrt{x^2-1}} \, dx,
 \qquad \Re t>0   \label{KNT}   \]
(see~\cite[7.12(19)]{BE}). One~has $K_0(t)\sim\log\frac1t$ as $t\searrow0$,
while
\[ K_0(t)\sim \sqrt{\frac\pi{2t}} \,e^{-t}
 \quad\text{as }t\to+\infty.   \label{tKE}   \]

\begin{lemma} \label{leBessel}
For any $k=0,1,2,\dots$,
\[ \int_0^\infty 2 t^k \, K_0(2\sqrt t) \, dt = k!^2.  \label{tKA}  \]
\end{lemma}

\begin{proof}
Making the indicated changes of variables and using Fubini,
\begin{align*}
& \int_0^\infty 2t^k \, K_0(2\sqrt t) \, dt =  \qquad(t\to\tfrac{t^2}4)  \\
&\qquad = \int_0^\infty t^{2k+1} 2^{-2k} K_0(t) \, dt  \\
&\qquad = \int_1^\infty \int_0^\infty t^{2k+1} 2^{-2k} \;
 \frac{e^{-tx}}{\sqrt{x^2-1}} \, dt \, dx \qquad(t\to\tfrac tx)   \\
&\qquad = \int_1^\infty 2^{-2k} x^{-2k-2} \frac{\Gamma(2k+2)}{\sqrt{x^2-1}}\,dx
 \qquad(x\to s^{-1/2})  \\
&\qquad = \int_0^1 2^{-2k} s^{k+1} \Gamma(2k+2) \sqrt{\frac s{1-s}}
 \; \frac{ds}{2s\sqrt s}  \\
&\qquad = 2^{-2k-1} \Gamma(2k+2) \frac{\Gamma(\frac12)k!}{\Gamma(k+\frac32)}
 = k!^2  \end{align*}
by~the doubling formula for the Gamma function.   \end{proof}

\begin{remark*}
Another way to arrive at (\ref{tKA}) is the following: starting with
the double integral
$$ k!^2 = \int_0^\infty \int_0^\infty e^{-x} e^{-y} x^k y^k \,dx \,dy ,  $$
we~make the change of variable $s=x+y$, $p=xy$, so $ds\,dp=|x-y|\,dx\,dy=
\sqrt{s^2-4p}\;dx\,dy$. This yields
\begin{align*}
k!^2 &= 2\int_0^\infty p^k \int_{2\sqrt p}^\infty e^{-s} \,
 \frac{ds}{\sqrt{s^2-4p}} \, dp    \\
&= 2\int_0^\infty p^k K_0(2\sqrt p) \, dp
\end{align*}
by~the change of variable $s\to2s\sqrt p$. Note that $2K_0(2\sqrt t)$ is~the
unique function whose moments are given by~(\ref{tKA}), in~view of Carleman's
criterion \cite[p.~85]{Akhi}, since $\sum_k k!^{-1/k}=\infty$ by Stirling's
formula.

Iterating the above argument, it~is also clear how to construct functions
on~$\RR_+$ whose moments will be~$k!^3$, or~$k!^4$, and so~forth.   \qed
\end{remark*}

Our~main result in this section is the following.

\begin{theorem} Let $0<\epsilon<1$ and denote for brevity
\[ c = \frac\epsilon{1-\epsilon}.   \label{tKF}  \]
Then for $m,n=0,1,2,\dots$,
\[ \intC L_n(z) \overline{L_m(z)} \, e^{cz+c\oz} \,
 K_0\Big(\frac{2\sqrt\epsilon}{1-\epsilon}|z|\Big) \, dz
 = \frac\pi{2c}\,\epsilon^{-n}\,\delta_{mn}.   \label{tKD}  \]
\end{theorem}

\begin{proof} By~virtue of the last lemma,
$$ \int_0^\infty t^k \, K_0\Big(\frac{2\sqrt\epsilon}{1-\epsilon}\sqrt t\Big)
 \,dt = \frac12 \Big(\frac{1-\epsilon}{\sqrt\epsilon}\Big)^{2k+2} k!^2.  $$
Using the generating function~(\ref{tLG}), we~thus have
\begin{align}
& \sum_{m,n=0}^\infty t^n s^m \intC L_n(z) \overline{L_m(z)} e^{cz+c\oz}
 K_0(\tfrac{2\sqrt\epsilon}{1-\epsilon}|z|) \, dz   \nonumber  \\
&\qquad = \frac1{(1-t)(1-s)} \intC
 e^{\frac{zt}{t-1}+\frac{\oz s}{s-1} + cz+c\oz} \,
 K_0(\tfrac{2\sqrt\epsilon}{1-\epsilon}|z|) \, dz   \nonumber  \\
&\qquad = \frac1{(1-t)(1-s)} \sum_{j,k=0}^\infty
 \Big(\frac t{t-1}+c\Big)^j \Big(\frac s{s-1}+c\Big)^k
 \intC \frac{z^j}{j!} \frac{\oz^k}{k!} \,
 K_0(\tfrac{2\sqrt\epsilon}{1-\epsilon}|z|) \, dz   \nonumber  \\
&\qquad = \frac1{(1-t)(1-s)} \sum_{j,k=0}^\infty
 \Big(\frac t{t-1}+c\Big)^j \Big(\frac s{s-1}+c\Big)^k
 \frac{\delta_{jk}\pi}{j!k!} \int_0^\infty t^k  \,
 K_0(\tfrac{2\sqrt\epsilon}{1-\epsilon} \sqrt t) \, dt   \nonumber  \\
&\qquad = \frac1{(1-t)(1-s)} \sum_{j,k=0}^\infty
 \Big(\frac t{t-1}+c\Big)^j \Big(\frac s{s-1}+c\Big)^k
 \frac\pi2 \delta_{jk} \Big(\frac{1-\epsilon}{\sqrt\epsilon}\Big)^{2k+2}
  \nonumber  \\
&\qquad = \frac{(1-\epsilon)^2\pi/(2\epsilon)}{(1-t)(1-s)}
 \Big[1-\frac{(1-\epsilon)^2}\epsilon \Big(\frac t{t-1}+c\Big)
 \Big(\frac s{s-1}+c\Big) \Big]^{-1},   \label{tKC}
\end{align}
where the interchange of the integration and summation in the first equality
is legitimate for
$$ \Big|\frac t{t-1}\Big|+c < \frac{\sqrt\epsilon}{1-\sqrt\epsilon}, \qquad
 \Big|\frac s{s-1}\Big|+c < \frac{\sqrt\epsilon}{1-\sqrt\epsilon}   $$
--- hence, for~$t,s$ in some neighbourhood of zero --- thanks to (\ref{tLE})
and~(\ref{tKE}). Now
\begin{align*}
& (1-t)(1-s)\Big[1-\frac{(1-\epsilon)^2}\epsilon \Big(\frac t{t-1}+c\Big)
 \Big(\frac s{s-1}+c\Big) \Big]  =  \\
&\qquad = (t-1)(s-1) - \frac{(1-\epsilon)^2}\epsilon(t+tc-c)(s+sc-c)  \\
&\qquad = \Big[1-c^2\frac{(1-\epsilon)^2}\epsilon\Big]
 + \Big[c(1+c)\frac{(1-\epsilon)^2}\epsilon-1\Big](t+s)
 + \Big[1-(1-c)^2\frac{(1-\epsilon)^2}\epsilon\Big] st   \\
&\qquad = (1-\epsilon) - \frac{1-\epsilon}\epsilon st   \\
&\qquad = (1-\epsilon) \Big(1-\frac{st}\epsilon\Big)
\end{align*}
by~(\ref{tKF}). Thus (\ref{tKC}) equals
$$ \frac{(1-\epsilon)\pi/(2\epsilon)} {1-\frac{ts}\epsilon}
 = \frac{(1-\epsilon)\pi}{2\epsilon}
 \sum_{k=0}^\infty \frac{(ts)^k}{\epsilon^k}  $$
and (\ref{tKD}) follows.   \end{proof}

\begin{corollary}
The~multiplication operator
$$ M_L: f(z) \longmapsto e^{\frac\epsilon{1-\epsilon}z} f(z)    $$
maps the space $\LL_\epsilon$ unitarily onto the space
$$ \LGe := L^2\hol(\CC,d\nu_\epsilon) $$
of entire functions on $\CC$ square-integrable with respect to the measure
\[ \postdisplaypenalty1000000
 d\nu_\epsilon(z) := \frac{2\epsilon}{(1-\epsilon)\pi} \,
 K_0\Big(\frac{2\sqrt\epsilon}{1-\epsilon}|z|\Big) \, dz ,  \label{NUE} \]
where $dz$ stands for the Lebesgue area measure.
\end{corollary}

The~space $\LGe=L^2\hol(\CC,d\nu_\epsilon)$ thus plays an analogous role
for the Laguerre polynomials as the Fock space $\FFe$ played for the
Hermite polynomials; we~will call $\LGe$ the \emph{Laguerre Fock space}.

The~last corollary and (\ref{tRPx}) imply that the reproducing kernel of
$\LGe$ is equal~to
$$ \frac1{1-\epsilon} I_0\Big(\frac{2\sqrt{x\oy\epsilon}}{1-\epsilon}\Big),  $$
which can be verified also directly using the monomial basis.
(Namely, quite generally, if a multiplication operator $M_\phi:f\mapsto\phi f$
is unitary from a reproducing kernel Hilbert space $\mathcal H_1$ into another
reproducing kernel Hilbert space~$\mathcal H_2$, then the corresponding
reproducing kernels are related~by $K_2(x,y)=\phi(x)K_1(x,y)\overline
{\phi(y)}$; this is immediate e.g.~from the standard formula $K(x,y)=\sum_j
e_j(x)\overline{e_j(y)}$ for reproducing kernel in terms of an arbitrary
orthonormal basis~$\{e_j\}$. As~for the second claim, Lemma~\ref{leBessel}
shows that $\{\frac{\epsilon^{j/2}}{(1-\epsilon)^{j+\frac12}}z^j\}_{j=0}
^\infty$ is an orthonormal basis in~$\LGe$, and the claim follows again by
the formula just mentioned.)

So~far we have worked with the ordinary Laguerre polynomials~$L_n(x)$;
it~should be noted, however, that everything we did in this section extends
in a routine manner also to the generalized Laguerre polynomials~$\La(x)$,
$\alpha>-1$, defined~by
$$ \La_n(x) = \frac{e^x x^{-\alpha}}{n!}\frac{d^n}{dx^n}x^{n+\alpha} e^{-x}. $$
They are orthogonal on the half-line $\RR_+=(0,+\infty)$ with respect to the
weight~$e^{-x}x^\alpha$
$$ \int_0^\infty \La_n(x) \La_m(x) \, e^{-x}x^\alpha \,dx =
 \frac{\Gamma(n+\alpha+1)}{n!} \delta_{mn},   $$
and can also be obtained from the generating function
$$ \sum_{n=0}^\infty \La_n(x) z^n = \frac1{(1-z)^{\alpha+1}}
 e^{\frac{xz}{z-1}}, \qquad x\in\CC, |z|<1 .    $$
The~ordinary Laguerre polynomials correspond to $\alpha=0$.
Similarly to our Lemma~4, one~checks that the modified Bessel functions
of the third kind
$$ K_\alpha(t) = \frac{\Gamma(\frac12)}{\Gamma(\frac12+\alpha)}
 \Big(\frac t2\Big)^\alpha
\int_1^\infty e^{-tx} (x^2-1)^{\alpha-\frac12} \, dx,
 \qquad \Re t>0, \; \alpha>-\tfrac12   $$
satisfies
\[ \int_0^\infty 2 t^{k+\frac\alpha2} \, K_\alpha(2\sqrt t) \, dt
 = k! \Gamma(k+\alpha+1).    \label{tQQ}  \]
{\allowdisplaybreaks[3]  \hfuzz50pt
With the $c$ from (\ref{tKF}), the computation
\begin{align*}
& \sum_{m,n=0}^\infty t^n s^m \intC \La_n(z) \overline{\La_m(z)} e^{cz+c\oz}
 |z|^\alpha K_\alpha(\tfrac{2\sqrt\epsilon}{1-\epsilon}|z|) \, dz  \\
&\qquad = \frac1{(1-t)^{\alpha+1}(1-s)^{\alpha+1}} \intC
 e^{\frac{zt}{t-1}+\frac{\oz s}{s-1} + cz+c\oz} \, |z|^\alpha
 K_\alpha(\tfrac{2\sqrt\epsilon}{1-\epsilon}|z|) \, dz   \\
&\qquad = \frac1{(1-t)^{\alpha+1}(1-s)^{\alpha+1}} \sum_{j,k=0}^\infty
 \Big(\frac t{t-1}+c\Big)^j \Big(\frac s{s-1}+c\Big)^k
 \intC \frac{z^j}{j!} \frac{\oz^k}{k!} \, |z|^\alpha
 K_\alpha(\tfrac{2\sqrt\epsilon}{1-\epsilon}|z|) \, dz   \\
&\qquad = \frac1{(1-t)^{\alpha+1}(1-s)^{\alpha+1}} \sum_{j,k=0}^\infty
 \Big(\frac t{t-1}+c\Big)^j \Big(\frac s{s-1}+c\Big)^k
 \frac{\delta_{jk}\pi}{j!k!} \int_0^\infty t^{k+\frac\alpha2}  \,
 K_\alpha(\tfrac{2\sqrt\epsilon}{1-\epsilon} \sqrt t) \, dt   \\
&\qquad = \frac1{(1-t)^{\alpha+1}(1-s)^{\alpha+1}} \sum_{j,k=0}^\infty
 \Big(\frac t{t-1}+c\Big)^j \Big(\frac s{s-1}+c\Big)^k
 \frac\pi2 \delta_{jk} \frac{\Gamma(k+\alpha+1)}{k!}
 \Big(\frac{1-\epsilon}{\sqrt\epsilon}\Big)^{2k+2+\alpha}  \\
&\qquad = \Big(\frac{1-\epsilon}{\sqrt\epsilon}\Big)^{2+\alpha}
 \frac{\Gamma(\alpha+1)\pi/2}{(1-t)^{\alpha+1}(1-s)^{\alpha+1}}
 \Big[1-\frac{(1-\epsilon)^2}\epsilon \Big(\frac t{t-1}+c\Big)
 \Big(\frac s{s-1}+c\Big) \Big]^{-\alpha-1}   \\
&\qquad = \Big(\frac{1-\epsilon}{\sqrt\epsilon}\Big)^{2+\alpha}
 \frac{\Gamma(\alpha+1)\pi}2
 (1-\epsilon)^{-\alpha-1} \Big(1-\frac{st}\epsilon\Big)^{-\alpha-1}   \\
&\qquad = \frac{(1-\epsilon)\Gamma(\alpha+1)\pi}
 {2\epsilon^{1+\frac\alpha2}(1-\frac{ts}\epsilon)^{\alpha+1}}   \\
&\qquad = \frac{(1-\epsilon)\pi} {2\epsilon^{1+\frac\alpha2}}
 \sum_{k=0}^\infty \frac{\Gamma(\alpha+k+1)}{k!} \frac{(ts)^k}{\epsilon^k}
\end{align*}
}shows as before that
$$ \intC \La_n(z) \overline{\La_m(z)} e^{cz+c\oz}
 |z|^\alpha K_\alpha(\tfrac{2\sqrt\epsilon}{1-\epsilon}|z|) \, dz
 = \frac{(1-\epsilon)\pi} {2\epsilon^{1+\frac\alpha2}}
 \frac{\Gamma(n+\alpha+1)}{n!} \delta_{mn} \epsilon^{-n} ,   $$
generalizing~(\ref{tKD}). The~same multiplication operator as before
$$ M_L: f(z) \longmapsto e^{\frac\epsilon{1-\epsilon}z} f(z)    $$
thus maps the space
$$ \LLa_\epsilon = \{ f=\sum_n f_n L_n: \sum_n \tfrac{\Gamma(n+\alpha+1)}{n!}
 \epsilon^{-n} |f_n|^2 =: \|f\|_{\LLa_\epsilon}^2 < \infty \}  $$
unitarily onto the space
$$ \LGea := L^2\hol(\CC,d\nua_\epsilon) $$
of entire functions on $\CC$ square-integrable with respect to the measure
$$ d\nua_\epsilon(z) := \frac{2\epsilon^{1+\frac\alpha2}}
 {(1-\epsilon)\pi} \, |z|^\alpha
 K_\alpha\Big(\frac{2\sqrt\epsilon}{1-\epsilon}|z|\Big) \, dz .   $$
The~reproducing kernel of $\LGea$ equals (cf.~\cite[\S10.12(20)]{BE})
$$ \frac1{1-\epsilon} (x\oy\epsilon)^{-\alpha/2}
I_\alpha \Big(\frac{2\sqrt{x\oy\epsilon}}{1-\epsilon}\Big),  $$
where $I_\alpha$ again denotes the modified Bessel function of the first kind.
In~particular, each $f\in\LLa_\epsilon$ again actually extends to an entire
function on~$\CC$, and $\LLa_\epsilon$ is the space of (the~restrictions
to~$\RR_+\!$~of) holomorphic functions on $\CC$ with reproducing kernel
$\frac1{1-\epsilon} e^{-\frac\epsilon{1-\epsilon}(x+\oy)}(x\oy\epsilon)
^{-\alpha/2} I_\alpha\Big(\frac{2\sqrt{\epsilon x\oy}}{1-\epsilon}\Big)$.

Remarkably, the~space $\LGea$ is a very well-known object, which first appeared
in Section~VI of the paper by Barut and Girardello \cite{BaG} on coherent
states associated with $SU(1,1)$; cf.~the~formulas (6.2) and (6.3) there.
(Our~$\alpha$ corresponds to $-2\Phi-1$ in the notation of~\cite{BaG};
recall that the Bessel function satisfies $K_\nu=K_{-\nu}$ for any~$\nu$.
Also we note that our Lemma~4 and (\ref{tQQ}) are just a special case of the
formula (3.26) there, however we have included the simple direct verification
here for convenience.) It~is noteworthy that Laguerre polynomials turn out
to be related to this space of Barut and Girardello in the same way as
Hermite polynomials were shown in \cite{AliKr} to be related to the
standard Fock-Segal-Bargmann space.
More recently, the space $\LGea$ has been studied in some detail
in~\cite{karp}. Another interesting point in this connection is the
existence of families of complex orthogonal polynomials in $z, \oz$,
with real coefficients, which span, for~example, the space $L^2(\CC,
d\mu_\epsilon)$, of~which $\FFe = L^2\hol(\CC, d\mu_\epsilon)$ is a subspace.
These polynomials are also known as complex Hermite polynomials
(see, e.g.,~\cite{ghan}), and~are determined completely by the 
measure~$d\mu_\epsilon$. A general procedure for constructing such
a family of polynomials, starting with a measure, has been developed
in~\cite{iszeng,iszhang}. It would be interesting to work out the
analogous complex orthogonal polynomials starting with the
measure~$d\nua_\epsilon$.

\section{The Laguerre ``squeeze'' operator}
In~\cite{AliKr}, it~was shown that the Hermite polynomial basis in the Fock
space actually arises as a ``squeezed'' variant of the standard monomial basis,
namely, the~former is obtained from the latter by a certain ``squeezing''
unitary operator. We~show that
all this persists also in the context of Laguerre polynomials and the Laguerre
Fock space of Barut and Girardello described in the preceding section.

For~simplicity, we~treat only the case $\alpha=0$, leaving the extension to
the generalized Laguerre polynomials $\La_n$ to the interested reader.

It is immediate from Lemma~4 that
$$ \Big\{ \frac{(-z)^n}{n!2^n\sqrt{2\pi}} \Big\} _{n=0}^\infty
 =: \{e_n(z)\}_{n=0}^\infty   $$
is an orthonormal basis of $L^2_{\text{hol}}(\CC,K_0(|z|)\,dz)$.

On~the other hand from Theorem~5, by~the~simple change of variable
$z\mapsto\frac{1-\epsilon}{2\sqrt\epsilon}z$, it~transpires that
$$ \Big\{ \sqrt{\frac{1-\epsilon}{2\pi}} \epsilon^{n/2} e^{\sqrt\epsilon z/2}
 L_n\Big(\frac{1-\epsilon}{2\sqrt\epsilon}z\Big) \Big\}_{n=0}^\infty
  =: \{E_{\epsilon,n}(z)\}_{n=0}^\infty  $$
is another orthonormal basis of the same space.

\begin{theorem}
Denote
$$ Q f (z) := \frac z2 f(z) - 2 \frac d{dz} z \frac d{dz} f(z).  $$
Then the operator
$$ U_\epsilon := \left(\sqrt{\frac{1+\sqrt\epsilon}{1-\sqrt\epsilon}}\right)
 {}\!{\vphantom{\bigg)}}^Q
 = \exp\Big[ \frac Q2 \log \frac{1+\sqrt\epsilon}{1-\sqrt\epsilon} \Big]  $$
satisfies
$$ U_\epsilon e_n = E_{\epsilon,n}, \qquad \forall n=0,1,2,\dots ,   $$
i.e.~maps the basis $\{e_n\}_{n=0}^\infty$ into the
basis~$\{E_{\epsilon,n}\}_{n=0}^\infty$.
\end{theorem}

Note that, by~a simple computation,
$$ \frac z2 e_n = -(n+1) e_{n+1}, \qquad
   2 \frac d{dz} z \frac d{dz} e_n = -n e_{n-1},   $$
from which one easily checks that $Q=T-T^*$ where $Tf(z)=\frac z2 f(z)$
is the operator of multiplication by $\frac z2$ on $L^2_{\text{hol}}
(\CC,K_0(|z|)\,dz)$. Thus $iQ$ is self-adjoint, and $U_\epsilon$ is unitary.
However, to~see that $U_\epsilon e_n=E_{\epsilon,n}$ requires more work.

\begin{proof} Recall once again the generating function for Legendre
polynomials
$$ (1-a) \sum_{n=0}^\infty a^n L_n(z) = e^{\frac a{a-1} z},
 \qquad |a|<1, z\in\CC .  $$
Taking in particular $a=\frac{w-1}{w+1}$ (so~$|a|<1$ corresponds to
$\Re w>0$), we~have $\frac a{a-1}=\frac{1-w}2$ and
$$ \frac2{w+1} \sum_{n=0}^\infty \Big(\frac{w-1}{w+1}\Big)^n L_n(z)
 = e^{\frac{1-w}2z},   $$
or
\[ \frac2{w+1} \sum_{n=0}^\infty \Big(\frac{w-1}{w+1}\Big)^n e^{-z/2} L_n(z)
 = e^{-wz/2}.  \label{tAA}    \]
On~the other hand, taking $a=\frac{w-1}{w+1}\frac{1+\epsilon}{1-\epsilon}$
(so~$|a|<1$ now corresponds to $w$ in the disc with diameter
$(\epsilon,\frac1\epsilon)$ in the right half-plane), we~similarly~get
\[ \frac{2(1-\epsilon w)}{(w+1)(1-\epsilon)} \sum_{n=0}^\infty
 \Big(\frac{w-1}{w+1}\frac{1+\epsilon}{1-\epsilon}\Big)^n e^{-z/2} L_n(z)
 = e^{\frac{\epsilon-w}{2(1-\epsilon w)}z}.   \label{tBB}   \]
Now from the differential equation for Legendre polynomials
$$ z L''_n(z) + (1-z) L'_n(z) = -n L_n(z)  $$
we~obtain upon a simple computation using just the Leibniz rule
$$ \Big(\frac z2-2\frac d{dz} z\frac d{dz} \Big) e^{-z/2} L_n(z)
 = (2n+1) e^{-z/2} L_n(z) ,   $$
i.e.~$Q(e^{-z/2}L_n(z))=(2n+1)e^{-z/2}L_n(z)$. Hence
$$ U_{\epsilon^2} (e^{-z/2}L_n(z)) = \Big(\frac{1+\epsilon}{1-\epsilon}\Big)
    ^{\frac{2n+1}2} e^{-z/2}L_n(z).  $$
Substituting this into (\ref{tAA}) yields
\begin{align*}
U_{\epsilon^2} e^{-wz/2} &= \frac2{w+1} \sqrt{\frac{1+\epsilon}{1-\epsilon}}
\sum_{n=0}^\infty \Big(\frac{w-1}{w+1} \frac{1+\epsilon}{1-\epsilon}\Big)^n
 e^{-z/2} L_n(z)   \\
&= \frac2{w+1} \sqrt{\frac{1+\epsilon}{1-\epsilon}}
 \frac{(w+1)(1-\epsilon)}{2(1-\epsilon w)}
 e^{\frac{\epsilon-w}{2(1-\epsilon w)}z}   \\
&= \frac{\sqrt{1-\epsilon^2}}{1-\epsilon w}
 e^{\frac{\epsilon-w}{2(1-\epsilon w)}z}   \end{align*}
by~(\ref{tBB}). Expanding the exponential on the left-hand side shows that it
equals
$$ U_{\epsilon^2} \sum_{n=0}^\infty \frac{(-z)^n}{n!2^n} w^n
 = \sqrt{2\pi} \sum_{n=0}^\infty w^n U_{\epsilon^2} e_n(z) .  $$
On~the other hand, using one more time the~generating function for Legendre
polynomials, this time with $a=\epsilon w$, shows that
\begin{align*}
\sqrt{2\pi} \sum_{n=0}^\infty w^n E_{\epsilon^2,n}(z)
&= \sum_{n=0}^\infty \sqrt{1-\epsilon^2} \epsilon^n w^n e^{\epsilon z/2}
 L_n\Big(\frac{1-\epsilon^2}{2\epsilon}z\Big)  \\
&= \sqrt{1-\epsilon^2} e^{\epsilon z/2} \frac1{1-\epsilon w}
 e^{\frac{\epsilon w}{\epsilon w-1} \frac{1-\epsilon^2}{2\epsilon} z} \\
&= \frac{\sqrt{1-\epsilon^2}}{1-\epsilon w}
 e^{\frac{\epsilon-w}{2(1-\epsilon w)}z} .   \end{align*}
Consequently,
$$ \sum_{n=0}^\infty w^n U_{\epsilon^2} e_n(z)
 = \sum_{n=0}^\infty w^n E_{\epsilon^2,n}(z).   $$
Comparing coefficients at like powers of $w$ and replacing $\epsilon$ by
$\sqrt\epsilon$, the theorem follows.    \end{proof}

Note that the operator $T$ and its adjoint $T^*$ mentioned before the last
proof coincide (up~to~a different normalization) with the generators
$\mathcal L_+$ and $\mathcal L_-$, respectively, of~the action of the Lie
algebra $\mathfrak{su}(1,1)$ on~$\LGe$ defined in (6.19) in~\cite{BaG}.
The~reader is referred to Section~V in \cite{AliKr} for further discussion
and physical interpretation of the ``squeezing'' procedure in the Hermite case.
By~analogy we shall refer to $U_\epsilon$ as the Laguerre
{\em squeeze operator}, although at this point we do not have a physical
meaning for this squeezing. Furthermore, using the squeeze operator,
we~could also derive a family of squeezed Barut-Girardello coherent states,
or~express the Barut-Girardello states themselves in terms of the
squeezed basis, just as was done for the canonical coherent
states in~\cite{AliKr}.

\section{Toeplitz operators on the Laguerre Fock space}
For a ``symbol'' $f\in L^\infty(\CC)$, the associated Toeplitz operator
$\Te_f=T_f$ on the Laguerre Fock space $\LGe$ is again given~by
\[ \Te_f u = P_\epsilon(fu), \qquad u\in\LGe,   \label{TXA}  \]
where $P_\epsilon:L^2(\CC,d\nu_\epsilon)\to\LGe$ is the orthogonal projection.
Our~aim in this section is to find the ``semi-classical'' asymptotics like
(\ref{tTA}) of these operators (with $h=1-\epsilon$).
There are well-established methods to handle this for measures $d\nu_\epsilon$
with power-like dependence on~$\epsilon$, that~is, of~the form
$d\nu_\epsilon(z)=F(z)^{c(\epsilon)}G(z)\,dz$ with some fixed positive
weights~$F,G$ and some real-valued function $c(\epsilon)$ of~$\epsilon$,
$c(\epsilon)\to+\infty$ as $\epsilon\nearrow1$; however, our $d\nu_\epsilon$
in (\ref{NUE}) are plainly not quite of this type,
so~we need to work from scratch.

Recall that, quite generally, on~a~family of reproducing kernel Hilbert spaces
$L^2\hol(\Omega,d\rho_\epsilon)$ of holomorphic functions with some measures
$d\rho_\epsilon$, $0<\epsilon<1$, on~a~domain $\Omega\in\CC^n$, establishing
an asymptotic expansion like (\ref{tTE}) for $T_fT_g$ is actually tantamount
to establishing the asymptotic behaviour of the \emph{Berezin transform}
$$ B_\epsilon f(z) := \int_\Omega f(w) \frac{|K_\epsilon(z,w)|^2}
 {K_\epsilon(z,z)} \, d\rho_\epsilon(w),   $$
where $K_\epsilon(z,w)$ is the reproducing kernel of
$L^2\hol(\Omega,d\rho_\epsilon)$. Indeed, from the definition (\ref{TXA})
it is immediate that $T_fT_g=T_{fg}$ whenever $g$ is holomorphic or
(upon taking adjoints) $f$ is anti-holomorphic.
Thus the bidifferential operators $C_j(f,g)$ in (\ref{tTE}) involve only
holomorphic derivatives of $f$ and anti-holomorphic derivatives of~$g$.
It~is therefore enough to determine $C_j(f,g)$ for holomorphic $f$ and
anti-holomorphic~$g$. For~such $f,g$, let~us apply both sides of (\ref{tTE})
to the reproducing kernel $K_{\epsilon,w}\equiv K_\epsilon(\cdot,w)$
at $w\in\Omega$, and evaluate at~$w$. Since $T_g K_{\epsilon,w}=g(w)
K_{\epsilon,w}$ for anti-holomorphic $g$ by the reproducong property
of~$K_\epsilon$, the left-hand side of (\ref{tTE}) gives just $f(w)g(w)
K_\epsilon(w,w)$; while the right-hand side, in~view of~(\ref{TXA}), becomes
$$ \sum_{j=0}^\infty h^j \int_\Omega C_j(f,g)(z) |K_\epsilon(z,w)|^2
 d\rho_\epsilon(z) = \sum_{j=0}^\infty h^j B_\epsilon [C_j(f,g)](w)
 K_\epsilon(w,w). $$
(Remember that $h=1-\epsilon$.) Consequently, we~get, at~least formally,
$$ \sum_{j=0}^\infty h^j C_j(f,g) = B_\epsilon^{-1}(fg),  $$
with the inverse being understood in the sense of formal power series in
$h=1-\epsilon$. In~other words, if~$B_\epsilon$ has an asymptotic expansion
\[ B_\epsilon \approx \sum_{j=0}^\infty (1-\epsilon)^j Q_j   \label{TXB}   \]
with some differential operators~$Q_j$, and
$$ B_\epsilon^{-1} \approx \sum_{j=0}^\infty (1-\epsilon)^j \cR_j, \qquad \cR_j =
 \sum_{\alpha,\beta} R_{j\alpha\beta} \partial^\alpha \dbar^\beta ,  $$
is~the inverse of (\ref{TXB}) (as~a~formal power series in $1-\epsilon$), then
\[ C_j(f,g) = \sum_{\alpha,\beta} R_{j\alpha\beta} (\partial^\alpha f)
 (\dbar^\beta g).   \label{TXO}   \]
(Here the summations extend over all multiindices $\alpha,\beta$.)

See \cite{E52} for more details of the above argument.

\begin{example}
For~the ordinary Fock space~$\mathcal F_h=L^2\hol(\CC,e^{-|x|^2/h}\frac{dx}
{\pi h})$, $h>0$, the reproducing kernel is given by $K_h(z,w)=e^{z\ow/h}$,~so
$$ B_h f(z) = \frac1{\pi h} \intC f(w) e^{-|z-w|^2/h} \,dw
 = e^{h\Delta/4} f(z)   $$
is just the heat solution operator at time $t=\frac h4$. Its formal inverse
$B_h^{-1}$ is thus $e^{-h\Delta/4}$, and
$$ C_j(f,g) = \frac{(-1)^j}{j!} (\partial^j f)(\dbar^j g) ,  $$
recovering the well-known formula for the Berezin-Toeplitz quantization
on~$\CC$. \end{example}

Returning to our Laguerre Fock space, we~are thus confronted with finding the
asymptotics as $\epsilon\nearrow1$ of the associated Berezin transform
\[ B_\epsilon f(z) = I_0\Big(\frac{2\sqrt\epsilon|z|}{1-\epsilon}\Big)^{-1}
 \intC f(w) \Big|I_0\Big(\frac{2\sqrt{\epsilon z\ow}}{1-\epsilon}\Big)\Big|
 ^2 \, \frac{d\nu_\epsilon(w)}{1-\epsilon} ,  \label{TXC}   \]
where we have used the formula for the reproducing kernel of $\LGe$ from the
end of Section~3.

It~turns out to be more convenient, instead of $\epsilon\in(0,1)$, to~use the
parameter
$$ \alpha := \frac{2\sqrt\epsilon}{1-\epsilon}.   $$
Thus $\epsilon\nearrow1$ corresponds to $\alpha\to+\infty$. We~will write
$B_\alpha$ instead of $B_\epsilon$ from now~on.

\begin{theorem}
Let $z\in\CC$, $z\neq0$. For any $f\in L^\infty(\CC)$ which is $C^\infty$ in a
neighbourhood of~$z$, we~have
\[ B_\alpha f(z) \approx \sum_{j=0}^\infty \alpha^{-j} Q_j f(z)
 \qquad\text{as } \alpha\to+\infty,  \label{TXD}   \]
for some differential operators $Q_j$ on $\CC\setminus\{0\}$, $j=0,1,2,\dots$
$($not~depending on~$f$ and~$z)$. Explicitly,
\[ \begin{aligned}
Q_0 &= I, \qquad Q_1=|z|\Delta,   \\
Q_1 &= \tfrac12\Delta + (z\partial+\oz\dbar)\Delta + \tfrac12\Delta^2.
 \end{aligned}  \label{TXK}   \]

For $z=0$ and $f\in L^\infty(\CC)$ smooth near the origin, we~have
\[ B_\alpha f(0) \approx \sum_{j=0}^\infty \alpha^{-2j} \Delta^j f(0)
 \qquad \text{as } \alpha\to+\infty .   \label{TXE}   \]
\end{theorem}

Note that the asymptotics are thus discontinuous at $z=0$; this can be viewed
as an analogue of the familiar Stokes phenomenon in complex analysis.

\begin{proof}
For $z=0$, (\ref{TXC}) becomes simply
\begin{align*}
 B_\alpha f(z) &= \frac1{1-\epsilon} \intC f(w) \, d\nu_\epsilon(w)
 = \frac{2\epsilon}{(1-\epsilon)^2} \intC f(w) K_0\Big(\frac{2\sqrt\epsilon}
 {1-\epsilon}|w|\Big) \, dw  \\
 &= \frac{\alpha^2}{2\pi} \intC f(w) K_0(\alpha|w|) \, dw  .
\end{align*}
For any $\delta>0$ and $\alpha\ge1$, we~have
\begin{align*}
& \Big| \int_{|w|>\delta} f(w) K_0(\alpha|w|) \,dw \Big|
 \le \|f\|_\infty 2\pi \int_\delta^\infty K_0(\alpha r) \, r\,dr   \\
&\hskip8em
 \le \|f\|_\infty c_\delta \int_\delta^\infty e^{-\alpha r/2} \,dr
 = \|f\|_\infty \frac{2c_\delta}\alpha e^{-\alpha\delta/2}
\end{align*}
for some finite $c_\delta$, thanks to~(\ref{tKE}). This decays faster than any
negative power of $\alpha$ as $\alpha\to+\infty$.

On~the other hand, for $|x|<\delta$ with $\delta$ small enough we may replace
$f$ by its Taylor expansion at the origin, giving
$$ \frac1{1-\epsilon} \sum_{j,k=0}^\infty \frac{\partial^j\dbar^k f(0)}{j!k!}
 \int_{|w|<\delta} w^j \ow^k \, d\nu_\epsilon(w).   $$
Again, modulo an exponentially small error, the last integral equals, as~we
have seen in Section~3,
$$ \intC w^j \ow^k \, d\nu_\epsilon(w) = \delta_{jk} \frac{(1-\epsilon)^{2j}}
 {\epsilon^j} j!^2 .  $$
Hence
\begin{align*}
B_\alpha f(0) &\approx \sum_{j,k=0}^\infty \frac{\partial^j\dbar^k f(0)}{j!k!}
 \delta_{jk} \frac{(1-\epsilon)^{2j}} {\epsilon^j} j!^2  \\
&= \sum_{k=0}^\infty \frac{(1-\epsilon)^{2k}\Delta^k f(0)}{4^k \epsilon^k}
 = \sum_{k=0}^\infty \frac{\Delta^k f(0)}{\alpha^{2k}}
\end{align*}
as $\alpha\to+\infty$, establishing~(\ref{TXE}).

For~the rest of the proof, we~thus assume $z\neq0$. The~change of variable
$$ w=(1+y)^2z, \qquad \Re y>-1,  $$
then transforms (\ref{TXC}) into
\begin{align*}
B_\alpha f(z) &= \frac1{I_0(\alpha|z|)} \int_{\Re y>-1} f((1+y)^2z) \,
 I_0((1+\oy)\alpha|z|) I_0((1+y)\alpha|z|) \\
&\hskip12em \times
 \frac{\alpha^2}{2\pi} K_0(\alpha|1+y|^2|z|) \, |2z(1+y)|^2 \, dy ,
\end{align*}
or, introducing
$$ \lambda := \alpha|z|    $$
for convenience,
\[ B_\alpha f(z) = \frac{2\lambda^2}{\pi I_0(\lambda)} \int_{\Re y>-1}
 f((1+y)^2z) \, |I_0((1+y)\lambda)|^2 K_0(|1+y|^2\lambda) \,|1+y|^2\,dy.
 \label{TXF}   \]
Using the integral representation
$$ I_0(z) = \frac1\pi \int_{-1}^1 \frac{e^{-tz}}{\sqrt{1-t^2}} \, dt  $$
for $I_0$ and the formula (\ref{KNT}) for~$K_0$, this can be rewritten~as
\begin{align*}
& \frac{\pi^3 I_0(\lambda)}{2\lambda^2} B_\alpha f(z) = \\
& \quad \int_{\Re y>-1} \int_{-1}^1 \int_{-1}^1 \int_1^\infty f((1+y)^2z)
 \frac{e^{-t\lambda(1+\oy)-s\lambda(1+y)-x\lambda|1+y|^2}}
 {\sqrt{(1-t^2)(1-s^2)(x^2-1)}} \, |1+y|^2 \, dx\,ds\,dt\,dy .
\end{align*}
Making one more change of variables
$$ t=T^2-1,\quad s=S^2-1,\quad x=X^2+1,  $$
the right-hand side becomes
\begin{align}
& \int_{\Re y>-1} \int_{-\sqrt2}^{\sqrt2} \int_{-\sqrt2}^{\sqrt2} \int_{\RR}
 \frac{f((1+y)^2z) |1+y|^2}{\sqrt{(2-T^2)(2-S^2)(2+X^2)}}   \nonumber \\
& \hskip8em {\vphantom\int} \times
 e^{(1-|y|^2-(1+\oy)T^2-(1+y)S^2-|1+y|^2X^2)\lambda} \, dX\,dS\,dT\,dy
 \nonumber   \\
& \begin{aligned} &\equiv e^\lambda \int_{\Re y>-1} \int_{-\sqrt2}^{\sqrt2}
 \int_{-\sqrt2}^{\sqrt2} \int_{\RR} F(y,T,S,X)  \\
 &\hskip8em \times
 e^{-(|y|^2+T^2+S^2+X^2+\cT(y,T,S,X))\lambda}    {\vphantom\int}
  \, dX\,dS\,dT\,dy,  \end{aligned}  \label{TXG}
\end{align}
where
\[ \begin{aligned}
F(y,T,S,X) &:= \frac{f((1+y)^2z) |1+y|^2}{\sqrt{(2-T^2)(2-S^2)(2+X^2)}} , \\
\cT(y,T,S,X)&:=T^2\oy + S^2 y+X^2(y+\oy+|y|^2) .  \end{aligned}  \label{TXI} \]
Note that the factor at $\lambda$ in the exponent in the integrand in
(\ref{TXG}) has a global maximum at the origin $T=S=X=y=0$, and vanishes
there precisely to the second order. Asymptotics as $\lambda\to+\infty$ of
such integrals is obtained by the standard stationary phase (WJKB) method;
in~the present case, this can be made quite explicit as follows.
Recall first of all that by the formula for the solution of the heat equation,
\begin{align*}
& \Big(\frac\lambda\pi\Big)^{5/2} \int_{\CC\times\RR^3} G(y,T,S,X)
 e^{-(|y|^2+T^2+S^2+X^2)\lambda} \, dX\,dS\,dT\,dy
 = e^{\Delta/(4\lambda)} G(0)   \\
&\hskip4em \approx \sum_{j=0}^\infty \frac1{j!(4\lambda)^j}
  \Big[ \frac{4\partial^2}{\partial y\partial\oy}
      + \frac{\partial^2}{\partial T^2} + \frac{\partial^2}{\partial S^2}
      + \frac{\partial^2}{\partial x^2} \Big]^j G(0)  .
\end{align*}
Arguing as in the beginning of this proof, one~sees that this holds also for
integration over any open subset containing the origin, instead of the whole
$\CC\times\RR^3$; in~particular, we~can apply it to the integral~(\ref{TXG}),
with
$$ G:= F\, e^{-\lambda\cT}  $$
(note that this depends also on~$\lambda$). We~thus obtain, at~least formally,
\begin{align*}
& \int\!\int\!\int\!\int F \, e^{-(|y|^2+T^2+S^2+X^2+\cT)\lambda}
 \, dX\,dS\,dT\,dy  \\
& \hskip1.2em \approx \Big(\frac\pi\lambda\Big)^{5/2} \sum_{j,k=0}^\infty
 \frac{\lambda^k}{j!(4\lambda)^j}
  \Big[ \frac{4\partial^2}{\partial y\partial\oy}
      + \frac{\partial^2}{\partial T^2} + \frac{\partial^2}{\partial S^2}
      + \frac{\partial^2}{\partial x^2} \Big]^j
  \Big( F\, \frac{(-\cT)^k}{k!} \Big) \;\Big|_{T=S=X=y=0} .
\end{align*}
Note that as $\cT$ vanishes to third order at the origin, one~gets nonzero
summands only for $0\le3k\le2j$. Rearranging the series we thus~get
\[ \begin{aligned}
& \Big(\frac\lambda\pi\Big)^{5/2} \frac{\pi^3 I_0(\lambda)}
 {2\lambda^2 e^\lambda} B_\alpha f(z) \approx   \\
&\hskip2em \sum_{m=0}^\infty \lambda^{-m}
 \sum_{j=0}^{2m} \Big[ \frac{\partial^2}{\partial y\partial\oy} + \frac14
  \Big(\frac{\partial^2}{\partial T^2}+\frac{\partial^2}{\partial S^2}
   +\frac{\partial^2}{\partial X^2}\Big)\Big]^{j+m} \frac{(-\cT)^j F}
  {(m+j)!j!} \Big|_{T=S=X=y=0} .   \end{aligned}   \label{TXH}  \]
Though our argument so far has been just formal, the~formula obtained is valid
and can be proved fully rigorously, see \cite[pp.~126--127]{Fed}.

Restoring $F$ and $\cT$ from~(\ref{TXI}), we~see that the right-hand side of
(\ref{TXH}) has the form
\[ \sum_{m=0}^\infty \lambda^{-m} (\cR_m f)(z),   \label{TXM}  \]
where $\cR_m$ are some differential operators on $\CC\setminus\{0\}$
with $C^\infty$ coefficients (in~fact, $\cR_m$~is of order~$2m$).
Explicit calculations (using computer for $m=2$) yield
\[ \begin{aligned}
& \cR_0 f = 2^{-3/2}f ,  \qquad
  \cR_1 f = 2^{-3/2} (\tfrac18 f+|z|\Delta f),  \\
& \cR_2 f = 2^{-3/2} (\tfrac9{128}f+\tfrac58|z|^2\Delta f
     +|z|^2(z\partial+\oz\dbar)\Delta f+\tfrac12|z|^2\Delta^2f) . \end{aligned}
   \label{TXL}   \]
Observe that in view of the reproducing property of the reproducing kernel,
one~has $B_\alpha\jedna=\jedna$ for all~$\alpha$ (where $\jedna$ denotes the
function constant~one). Consequently, taking $f=\jedna$ in~(\ref{TXH}),
\[ \Big(\frac\lambda\pi\Big)^{5/2} \frac{\pi^3 I_0(\lambda)}
 {2\lambda^2 e^\lambda} \approx \sum_{m=0}^\infty \lambda^{-m}
 (\cR_m\jedna)(z) .   \label{TXJ}   \]
Dividing (\ref{TXH}) by~(\ref{TXJ}), we~finally obtain
$$ B_\alpha f(z) \approx
 \frac{2\sqrt2 \sum_{m=0}^\infty \lambda^{-m} (\cR_m f)(z)}
      {2\sqrt2 \sum_{m=0}^\infty \lambda^{-m} (\cR_m \jedna)(z)}
  =: \sum_{m=0}^\infty \lambda^{-m} (\mathcal Q_m f)(z)    $$
with some differential operators $\mathcal Q_j$ on $\CC\setminus\{0\}$,
proving~(\ref{TXD}). (Note that the division of formal power series above
makes sense, since $2\sqrt2\cR_0\jedna=\jedna$.) Finally, lengthy but routine
calculation using (\ref{TXL}) yields~(\ref{TXK}).
\end{proof}

\begin{remark*}
A~somewhat simpler~way (which would require some justification however to make
it completely rigorous) to~get explicit expressions for the $\cR_m$ and $Q_m$
above is as follows. Recall that as $z\to\infty$, the functions $I_0$ and $K_0$
possess the asymptotic expansions
\begin{align}
I_0(z) &\approx \frac{e^z}{\sqrt{2\pi z}} \sum_{m=0}^\infty \frac{c_m}{z^m},
  \label{TXN}  \\
K_0(z) &\approx \sqrt{\frac\pi{2z}} e^{-z} \sum_{m=0}^\infty \frac{(-1)^mc_m}
  {z^m},   \label{TXQ}
\end{align}
where
$$ c_m = \frac{(\frac12)_m^2}{m!2^m} = \frac{\Gamma(m+\frac12)^2}{m!2^m\pi}; $$
see \cite[vol.~II, \S7.13.1, (5)~and~(7)]{BE}. Substituting these
into~(\ref{TXF}) yields
\begin{align*}
\frac{\pi I_0(\lambda)}{2\lambda^2} B_\alpha(z)
&= \sum_{j,k,l=0}^\infty \frac{c_jc_k(-1)^lc_l}{\lambda^{j+k+l+3/2}}
    e^\lambda \int_{\Re y>-1} \frac{f((1+y)^2z)}{(1+y)^j(1+\oy)^k|1+y|^{2l}}
    e^{-\lambda|y|^2}\, dy   \\
&\approx \frac{\sqrt\pi e^\lambda}{\sqrt8 \lambda^{5/2}}
  \sum_{j,k,l,n=0}^\infty
  \frac{c_jc_k(-1)^lc_l}{n!\lambda^{j+k+l+n}} \,\partial^n\dbar^n
  \Big[ \frac{f((1+y)^2z)} {(1+y)^{j+l}(1+\oy)^{k+l}} \Big] _{y=0}.
\end{align*}
Comparing this with (\ref{TXM}) yields immediately
$$ \cR_m f(z) = \sum_{j+k+l+n=m}
  \frac{c_jc_k(-1)^lc_l}{n!\sqrt8} \,\partial^n\dbar^n
  \Big[ \frac{f((1+y)^2z)} {(1+y)^{j+l}(1+\oy)^{k+l}} \Big] _{y=0},   $$
which is a much simpler expression than in~(\ref{TXH}).

We~pause to note that it is amusing to check that the asymptotics of
$I_0(\lambda)$ implicit from (\ref{TXJ}) coincide with~(\ref{TXN}).  \qed
\end{remark*}

Returning to our Toeplitz operator asymptotics on the Laguerre Fock space,
we~see from (\ref{TXK}) and (\ref{TXO}) that
$$ T_f T_g-T_g T_f \approx \frac{ih}{2\pi} T_{\{f,g\}},   $$
where
$$ \frac{ih}{2\pi}\{f,g\} = 2\frac{1-\epsilon}{\sqrt\epsilon}|z|
    (\partial f \dbar g- \dbar f\partial g) + O(h^2)   $$
(recall that $h=1-\epsilon$). Thus what we have is a quantization of the
K\"ahler metric
\[ \frac{dz\,d\oz}{2|z|} .  \label{TXP}   \]
In~view of the singularity at $z=0$, we~are in effect quantizing not $\CC$
but $\CC\setminus\{0\}$, where (\ref{TXP}) is just the pullback of the
(appropriately rescaled) Euclidean metricc on the universal cover $\CC$ of
$\CC\setminus\{0\}$. (This accounts for the discontinuity of the asymptotics
at $z=0$: physically, the~origin does not belong to our phase space and the
asymptotics there have no physical relevance.) A~potential for this metric
is given by $\Psi(z)=2|z|$, so~the traditional Berezin-Toeplitz quantization
would be using the spaces
$$ L^2\hol(\CC\setminus\{0\},e^{-2|z|/h}dz) = L^2\hol(\CC,e^{-2|z|/h}dz),
 \qquad h>0,  $$
as~described in the Introduction. (The~equality of the last two spaces follows
from the well-known fact --- easily checked using the Laurent expansion and
polar coordinates --- that any holomorphic and square-integrable function in
a punctured neighbourhood of the origin has a removable singularity there.)
The~latter can be carried out as in Example~2.16 in \cite{E12}, and we leave
to the reader the (amusing) comparison of the outcomes of the two approaches.

We~conclude this section by mentioning that, analogously as in Section~6
in~\cite{AEH}, one~can in principle derive the asymptotics of the Toeplitz
operators on $\LGe$ also by using the standard Weyl calculus on $L^2(\RR)$.
Namely, the~integral operator
$$ V_L f(z) := \int_0^\infty f(x) \beta_L(z,x) \, dx  $$
where
\begin{align*}
\beta_L(z,x)
&= \sum_n l_n(x) \epsilon^{n/2} L_n(z) e^{\epsilon z/(1-\epsilon)}  \\
&= e^{\frac\epsilon{1-\epsilon}z - \frac x2} L_{\sqrt\epsilon}(z,x)  \\
&= \frac1{1-\sqrt\epsilon}
 e^{-\frac{\sqrt\epsilon}{1-\epsilon}z
    - \frac{1+\sqrt\epsilon}{2(1-\sqrt\epsilon)}x}
 I_0\Big( \frac{2\sqrt{xz}\epsilon^{1/4}}{1-\sqrt\epsilon} \Big)
\end{align*}
is~a~unitary isomorphism of $L^2(\RR_+)$ onto $L^2\hol(\CC,d\nu_\epsilon)$
(taking the orthonormal basis $\{l_n\}_n$ of $L^2(\RR_+)$ into the orthonormal
basis $\{\epsilon^{n/2} e^{\epsilon z/(1-\epsilon)} L_n(z)\}_n$ of the latter).
Composing it with the obvious unitary isomorphism
$$ Q: f(x) \longmapsto x^{-1/2} f(\log x)   $$
of $L^2(\RR)$ onto~$L^2(\RR_+)$, we~thus obtain the operator
$$ V_L Q f(z) = \intR f(x) \beta_L(z,e^x) e^{x/2} \,dx   $$
sending $L^2(\RR)$ unitarily onto~$\LGe$.

One~can now, in~principle, again consider the Toeplitz operators~$T_\phi$,
$\phi\in L^\infty(\CC)$, on~$L^2\hol(\CC,d\nu_\epsilon)$ and try to identify
the transferred operator $Q^*V^*_L T_\phi V_L Q$ with some Weyl operator
on~$L^2(\RR)$. Proceeding as in Section~6 in~\cite{AEH}, we~find that
$Q^*V^*_LT_\phi V_LQ=W_a$ with $a$ given~by
\begin{align*}
& \check a\Big(\frac{x+y}2,x-y\Big)
= \intC \phi(z) \beta_L(z,e^y) \overline{\beta_L(z,e^x)} \, e^{\frac{x+y}2}
 \, d\nu_\epsilon(z)   \\
&\qquad = \frac{2\epsilon}{(1-\epsilon)(1-\sqrt\epsilon)^2\pi}
 e^{\frac x2 - \frac{1+\sqrt\epsilon}{2(1-\sqrt\epsilon)} e^x
    + \frac y2 - \frac{1+\sqrt\epsilon}{2(1-\sqrt\epsilon)} e^y}  \\
&\hskip4em \intC \phi(z) \, e^{-\frac{\sqrt\epsilon}{1-\epsilon}(z+\oz)}
 I_0\Big( \frac{2\sqrt z e^{y/2} \epsilon^{1/4}}{1-\sqrt\epsilon} \Big)
 I_0\Big( \frac{2\sqrt{\oz} e^{x/2} \epsilon^{1/4}}{1-\sqrt\epsilon} \Big)
 K_0\Big(\frac{2\sqrt\epsilon}{1-\epsilon}|z|\Big) \, dz ,
\end{align*}
whence
\begin{align*}
a(s,\eta) &= \frac{2\epsilon}{(1-\epsilon)(1-\sqrt\epsilon)^2\pi}
 \intC \phi(z) \, K_0\Big(\frac{2\sqrt\epsilon}{1-\epsilon}|z|\Big) \,
 e^{-\frac{\sqrt\epsilon}{1-\epsilon}(z+\oz) +s}   \\
&\qquad \intR e^{-ir\eta}
 e^{-\frac{1+\sqrt\epsilon}{1-\sqrt\epsilon} \cosh\frac r2} \,
 I_0\Big(\frac{2\sqrt{z\epsilon^{1/2}e^s}}{1-\sqrt\epsilon} e^{-\frac r4}\Big)
 I_0\Big(\frac{2\sqrt{\oz\epsilon^{1/2}e^s}}{1-\sqrt\epsilon} e^{\frac r4}\Big)
 \, dr \, dz .
\end{align*}
One~can now again replace $I_0$ and $K_0$ by their integral representations
(or,~at~least on a heuristic level, by~their asymptotic expansions (\ref{TXN})
and~(\ref{TXQ})) and proceed as before to obtain an asymptotic expansion for
$a(s,\eta)$ as $\epsilon\nearrow1$. Invoking the standard composition rules
for the Weyl calculus would then lead to the asymptotics (\ref{tTE}) of the
Toeplitz product~$T_fT_g$. We~omit the details.

\section{Legendre polynomials}
Another family of orthogonal polynomials susceptible to a similar treatment
as with $H_n(x)$ and $L_n(x)$ are the Legendre polynomials $P_n(x)$,
$n=0,1,2,\dots$, defined~by
$$ P_n(x) := \frac{(-1)^n}{n!2^n} \frac{d^n}{dx^n} (1-x^2)^n.   $$
These polynomials form an orthogonal basis on $L^2(-1,1)$:
$$ \int_{-1}^1 P_n(x) P_m(x) \,dx = \frac{\delta_{mn}}{m+\frac12}.   $$
The~corresponding series
\[ \tPe(x,y) = \sum_{n=0}^\infty \epsilon^n
 \Big(n+\frac12\Big)P_n(x)P_n(y), \qquad x,y\in(-1,1),  \label{tRS}   \]
can~be summed to the rather complicated expression
\begin{align*}
 & \tPe(\cos2\phi,\cos2\theta) = \\
 & \qquad \frac{1-\epsilon}{2(1+\epsilon)^2}
 \sum_{m,n=0}^\infty \frac{(m+n)!(\frac32)_{m+n}} {(m!n!)^2} \frac
 {(4\epsilon\sin^2\phi\sin^2\theta)^m (4\epsilon\cos^2\phi\cos^2\theta)^n}
 {(1+\epsilon)^{2m+2n}} ,   \end{align*}
see~\cite[(7.5.6)]{Askey}.
(Incidentally, the series like $\tKe,\tLe$ and $\tPe$ are called the ``Poisson
kernels'' for the corresponding orthogonal polynomials in~\cite{Askey}.
The series on the right-hand side in the last formula is Appell's
hypergeometric function~$F_4$ in Horn's notation~\cite[\S5.7.1]{BE}.)

The~differential equation~is
$$ (1-x^2)P''_n(x) - 2x P'_n(x) + n(n+1) P_n(x) = 0 ,   $$
implying that
$$ AP_n=nP_n  $$
for
\[ A = \sqrt{-D(1-x^2)D+\frac14I}-\frac12I   \label{tRO}   \]
($D=d/dx$, and the square root is understood in the spectral-theoretic sense).
As~before, it~follows that the corresponding ``Toeplitz'' operators
$$ \Te_f u(x) = \int_{-1}^1 u(y) f(y) \tPe(x,y) \, dy, \qquad 0<\epsilon<1,  $$
satisfy~(\ref{tRJ}), with the operator~$A$ from~(\ref{tRO}). There are also the
corresponding Hilbert spaces $\PP_\epsilon$ of functions on $(-1,1)$
having $\tPe$ for their reproducing kernel; however, unlike the situation for
Hermite and Laguerre polynomials, in~this case $\PP_\epsilon$ no longer
extend to~a~reproducing kernel Hilbert space on a larger~set.

\begin{proposition}
There exists no domain $\Omega$ in~$\CC$ containing the interval $(-1,1)$ and
such that for each $0<\epsilon<1$, $\PP_\epsilon$ would consist of
restrictions to $(-1,1)$ of functions in some reproducing kernel Hilbert space
of holomorphic functions on~$\Omega$.   \end{proposition}

\begin{proof}
Assume, to~the contrary, that such a domain $\Omega$ and reproducing kernel
Hilbert spaces $\Peo$, $0<\epsilon<1$, exist. For~each~$\epsilon$, the function
$\tPe(x,y)$ then extends to a function (still denoted~$\tPe$) on~all of
$\Omega\times\Omega$, holomorphic in~$x,\oy$, which is the reproducing
kernel of~$\Peo$; furthermore, by~the standard formula for the reproducing
kernel in terms of an orthonormal basis, the~series (\ref{tRS}) converges
for any $x,y\in\Omega$. Thus, in~particular, the~series
$$ \sum_{n=0}^\infty (n+\tfrac12)\epsilon^n|P_n(x)|^2   $$
converges for any $\epsilon\in(0,1)$ and $x\in\Omega$.
By~the Cauchy-Schwarz inequality
$$ \Big(\sum_n|z^n P_n(x)|\Big)^2 \le \Big(\sum_n|z|^n\Big)
 \Big(\sum_n|z|^n|P_n(x)|^2\Big) \le \frac2{1-|z|}
 \sum_n(n+\tfrac12)|z|^n|P_n(x)|^2   $$
it~thus follows that for any $x\in\Omega$, the~series $\sum_n z^n P_n(x)$
converges for any $|z|<1$. However, using the familiar generating function
for Legendre polynomials
$$ \sum_{n=0}^\infty z^n P_n(x) = (1-2xz+z^2)^{-1/2},
 \qquad |z|<1, \; -1\le x\le 1,   $$
we~quickly see that the series on the left-hand side converges precisely~for
$$ |z| < \min(|x+\sqrt{x^2-1}|,|x-\sqrt{x^2-1}|) .   $$
Since
$$ |x+\sqrt{x^2-1}|\cdot|x-\sqrt{x^2-1}|=1 ,  $$
the~series can thus converge for all $|z|<1$ only if both $x+\sqrt{x^2-1}$ and
$x-\sqrt{x^2-1}$ lie on the unit circle, that~is, if~and only if $x\in[-1,1]$.
\end{proof}

\begin{remark*}
For~a~given $\epsilon\in(0,1)$, the~domain of convergence of the series
(\ref{tRS}) is given~by (cf.~\cite[\S5.7]{BE})
$$ |(1-x)(1-y)|^{1/2}+|(1+x)(1+y)|^{1/2} < \epsilon^{1/2}+\epsilon^{-1/2}.  $$
Thus $\PP_\epsilon$ actually extends to a reproducing kernel Hilbert
space of holomorphic functions on the ellipse
$$ \Omega_\epsilon := \{x\in\CC: \; |1-x|+|1+x|<(1+\epsilon)/\sqrt\epsilon\} $$
which however shrinks to the interval $[-1,1]$ as $\epsilon\nearrow1$.
\end{remark*}

A~more explicit description of the space $\PP_\epsilon$ for a given $\epsilon$
was given in~\cite{karp}. One~can treat in the same way also the Jacobi
polynomials $P_n^{(\alpha,\beta)}$, $\alpha,\beta>-1$ (of~which $P_n$ are
the special case $\alpha=\beta=0$); we~omit the details.

\section{Final remarks: other sequences}
The~choice of the powers $\epsilon^n$ in (\ref{tTF}), (\ref{tRP}) and
(\ref{tRS}) may admittedly seem rather haphazard. It~is in fact possible
to give a fairly complete picture of what happens, from the point of view
of existence of the reproducing kernel Hilbert spaces like $\HHe$,
$\LL_\epsilon$ and~$\PP_\epsilon$, when it is replaced by other
sequences of positive coefficients.

\begin{theorem}
Let $c_n$ be a sequence of positive numbers. Then the following are equivalent:
\begin{enumerate}
\item[(a)] The series $P_c(x,y):=\sum_n c_n P_n(x)P_n(y)$ converges for all
$x,y\in[-1,1]$ and $P_c(x,y)$ is the reproducing kernel of the Hilbert space
$$ \PP_c := \{ f=\sum_n f_n P_n: \;
 \sum_n c_n^{-1}|f_n|^2 =: \|f\|_{\PP_c}^2 <\infty \}  $$
of functions on~$[-1,1]$.
\item[(b)] $\sum_n c_n<\infty$.
\end{enumerate}
\end{theorem}

\begin{proof}
(a)$\implies$(b) This is immediate upon taking $x=y=1$, since $P_n(1)=1$
$\forall n$.

(b)$\implies$(a) Since
\[ |P_n(x)| \le 1 \qquad \forall n \; \forall x\in[-1,1]  \label{tPJ}  \]
(this follows e.g.~from the first formula in~\cite[10.10(42)]{BE}),
clearly $P_c(x,y)$ converges for all $x,y\in[-1,1]$, and thus
$P_{c,y}:=P_c(\cdot,y)$ belongs to $\PP_c$ for each $y\in[-1,1]$.
The~rest of the argument is the same as in the proof of
Proposition~1 in~\cite{AEH}: namely, for~any $f\in\PP_c$,
we~have using again~(\ref{tPJ})
$$ \sum_n|f_n P_n(y)| \le \Big(\sum_n c_n^{-1}|f_n|^2\Big)^{1/2}
 \Big(\sum_n c_n|P_n(y)|^2\Big)^{1/2}
 = \|f\|_{\PP_c} P_c(y,y)^{1/2} < \infty  $$
showing that the series $\sum_n f_n P_n(y)=:f(y)$ converges and
$f\mapsto f(y)=\spr{f,P_{c,y}}_{\PP_c}$ is~a~bounded linear
functional~on~$\PP_c$. Thus $\PP_c$ is a reproducing kernel
Hilbert space with reproducing kernel $P_c(x,y)$, as~asserted.
\end{proof}

\begin{theorem}
Let $c_n$ be a sequence of positive numbers. Then the following are equivalent:
\begin{enumerate}
\item[(a)] The series $L_c(x,y):=\sum_n c_n L_n(x)L_n(y)$ converges for all
$x,y\ge0$ and $L_c(x,y)$ is the reproducing kernel of the Hilbert space
$$ \LL_c := \{ f=\sum_n f_n L_n: \;
 \sum_n c_n^{-1}|f_n|^2 =: \|f\|_{\LL_c}^2 <\infty \}  $$   
of functions on~$[0,\infty)$.  
\item[(b)] $\sum_n c_n<\infty$.
\end{enumerate}
\end{theorem}

\begin{proof}
(a)$\implies$(b) Immediate upon taking $x=y=0$, since $L_n(0)=1$ $\forall n$.

(b)$\implies$(a) Recall that the Legendre polynomials are related
by the formula
$$ L_n(x) = \frac{(-1)^n}{n!} \Psi(-n,1,x), \qquad x\neq0,   $$
to~the confluent hypergeometric function~$\Psi$ \cite[10.12(14)]{BE}.
The~latter possesses the asymptotic behaviour \cite[6.13(8)]{BE}
\[ \Psi(a,c,x) = \kappa^{\kappa-\frac14} x^{\frac c2-\frac14}
 e^{\frac x2-\kappa} \sqrt2 \cos(\kappa\pi-2\sqrt{\kappa x}-\tfrac\pi4)
 \cdot[1+O(|\kappa|^{-1/2})]
 \label{tPS}  \]
as $\kappa:=\tfrac c2-a\to+\infty$.
For $x>0$, the~cosine is bounded by 1 in modulus, thus by Stirling's formula
$$ |L_n(x)| \le C_x n^{-1/4}   $$
for all $n$ large enough, and the convergence of $L_c(x,y)$ follows.
The~rest of the proof is the same as for the preceding theorem.
\end{proof}

\begin{remark*} Proceeding as in the proof of Theorem~\ref{thmLag}, one~can
show that $L_c(x,y)$ in fact converges for all $x,y\in\CC$, and $\LL_c$ extends
to a space of holomorphic functions on all of~$\CC$, as~soon as $\sum_n c_n
r^{-2n}<\infty$ for some $r\in(0,1)$. The~latter condition can in fact be
relaxed~to
\[ \sum_n c_n e^{a\sqrt n} < \infty \qquad \forall a>0  \label{tCN}  \]
(or,~equivalently, $c_n^{1/\sqrt n}\to0$) by~using~(\ref{tPS}).   \end{remark*}

For~the spaces of Hermite polynomials, it~is easy to see that $\sum_n c_n
\frac{H_n(x) H_n(y)}{n!2^n\sqrt\pi}$ converges for $x=y=0$ if and only if
$\sum_{n=1}^\infty c_{2n}/\sqrt n<\infty$; unfortunately, handling the $c_n$
with odd $n$ seems more difficult. We~can however offer the following result.

\begin{theorem}
Let $c_n$ be a sequence of positive numbers. Then the following are equivalent:
\begin{enumerate}
\item[(a)] The two series $H_c(x,y):=\sum_n c_n H_n(x)H_n(y)
(n!2^n\sqrt\pi)^{-1}$ and $H^\#_c(x,y):=\sum_n c_{n+1} H_n(x)H_n(y)
(n!2^n\sqrt\pi)^{-1}$ converge for all $x,y\in\RR$ and $H_c(x,y)$
is the reproducing kernel of the Hilbert space
$$ \HH_c := \{ f=\sum_n f_n (n!2^n\sqrt\pi)^{-1/2} H_n: \;
 \sum_n c_n^{-1}|f_n|^2 =: \|f\|_{\HH_c}^2 <\infty \}  $$
of functions on~$\RR$.
\item[(b)] $\sum_{n=1}^\infty n^{-1/2}c_n<\infty$.
\end{enumerate}
\end{theorem}

\begin{proof}
(a)$\implies$(b) As~already mentioned, taking $x=y=0$, $H_{2n+1}(0)=0$ and
$H_{2n}(0)=\frac{(2n)!(-1)^n}{n!}$ imply that
$$ \infty > \sum_n c_{2n} \frac{(2n)!}{n!^2 2^{2n}}
 = \sum_n c_{2n} \frac{(\frac12)_n}{n!} \sim \sum_n c_{2n} n^{-1/2} ;  $$
the same argument for $H^\#_c$ gives $\sum_n c_{2n+1}n^{-1/2}<\infty$,
and (b) follows.

(b)$\implies$(a) According to a result of Schwid~\cite[Theorem~VIII(a)]{Schwi}
and Stirling's formula,
$$ \frac{H_n(z)}{\sqrt{n!2^n\pi^{1/2}}} =
 2^{1/4} \pi^{-1/2} e^{z^2/2} n^{-1/4} [1+O(n^{-1})]
 \Big[\cos\Big(\frac{\pi n}2-z\sqrt{2n+1}\Big)+O(n^{-1/2})\Big]   $$
as $n\to+\infty$. For $z$ real, the cosine is bounded,~so
$$ \Big| \frac{H_n(z)}{\sqrt{n!2^n\pi^{1/2}}} \Big| \le C_z n^{-1/4}
 \qquad\text{ for $n$ large enough,}   $$
and the convergence of $H_c(x,y)$ for any $x,y\in\RR$ follows.
The~assertion for $H^\#_c(x,y)$ is obtained upon replacing
$\{c_n\}$ by~$\{c_{n+1}\}$.   \end{proof}

\begin{remark*}
It~follows from the proof of Theorem~2 in \cite{AEH} that, again,
$H_c(x,y)$ in fact converges for all $x,y\in\CC$, and $\HH_c$
extends to a space of holomorphic functions on all of~$\CC$,
as~soon as the sequence $\{c_n\}$ satisfies~(\ref{tCN}).

For~Legendre polynomials, the~condition for $P_c(x,y)$ to converge for
all $x,y\in\CC$, and for $\PP_c$ to extend to a reproducing kernel
Hilbert space of holomorphic functions on all of~$\CC$, can~similarly
be shown to be $c_n^{1/n}\to0$. 
\end{remark*}

\section{Conclusion}
Generally, if we we start with a family of real polynomials $p_n(x), \;\;x \in
\mathbb R, \;\; n =0,1,2, \ldots, \infty$, which are orthonormal with respect
to a measure $d\mu$ over $\mathbb R$, the sum $\sum_{n=0}^\infty p_n(x)p_n(y)$
is usually divergent. However, there exist families of  polynomials, such as
the ones considered in this paper, for which the sum $K_\epsilon (x, y) =
\sum_{n=0}^\infty \epsilon^n p_n(x)p_n(y), \;\; 0<\epsilon <1$, coverges for
all $x,y$. In that case $K_\epsilon (x,y)$ defines a reproducing kernel and the
polynomials $\epsilon^{\frac n2}p_n (x)$ constitute an orthonormal basis for
the corresponding Hilbert space $\mathfrak H_\epsilon$. However, although
$\mathfrak H_\epsilon \subset L^2 (\mathbb R, d\mu)$, it is in general not a
Hilbert subspace. On the other hand, if we write the same polynomials in a
complex variable, $\epsilon^{\frac n2}p_n(z), \;\; z \in \mathbb C$, it often
turns out that the sum $K_\epsilon (z, z') = \sum_{n=0}^\infty \epsilon^n
p_n(z)p_n(z'), \;\; z, z' \in \mathbb C$, is convergent in some domain of the
complex plane, in which it defines a reproducing kernel. Moreover, the
corresponding reproducing kernel Hilbert space turns out to be a (holomorphic)
subspace of an $L^2$-space over this domain. This is the general situation
which is known to happen, for example, for the Hermite, Laguerre and Jacobi
polynomials. Additionally, a large number of other interesting questions
emerge, related to such families of polynomials and reproducing kernel Hilbert
spaces. In this paper and in \cite{AEH}, we have looked at the questions of
Berezin-Toeplitz quantization using the real kernel $K_\epsilon (x,y)$ and its
semi-classical approximation, and to certain physical questions related to
``squeezing'' of coherent states. In a future publication we plan to look at
the problems of the associated non-linear coherent states and complex
orthogonal polynomials related to such systems.

\section*{Acknowledgements}
The work reported here was partly supported by a GA\v CR grant no.~201/12/0426,
RVO funding for I\v CO~67985840 and the Natural Sciences and Engineering
Research Council (NSERC) of Canada. 
Part of this work was done while the second author was visiting the first in
September~2014; the hospitality of the Department of Mathematics and Statistics
of Concordia University on that occasion is gratefully acknowledged.


\begin{thebibliography}{99}
\bibitem{Ake} G. Akemann, {\it Complex Laguerre symplectic ensemble of
non-Hermitian matrices,\/} Nucl. Phys. B {\bf 730} (2005), 253--299.

\bibitem{Akhi} N.I. Akhiezer, {\it The classical moment problem,\/}
Oliver \& Boyd, London, 1965.

\bibitem{AEH} S. T. Ali, M. Engli\v s: {\it Hermite polynomials and
quasi-classical asymptotics,\/} J.~Math. Phys. {\bf 55} (2014), 042102.

\bibitem{AliKr} S.T. Ali, K.~Gorska, A.~Horzela, F.H.~Szafraniec:
{\it Squeezed states and Hermite polynomials in a complex variable,\/}
J. Math. Phys. {\bf 55} (2014), 012107 (11~pp).

\bibitem{Askey} G.E. Andrews, R. Askey, R. Roy, {\it Special functions,\/}
Cambridge University Press, Cambridge, 1999.

\bibitem{Aro} N. Aronszajn: {\it Theory of reproducing kernels\/},
Trans. Amer. Math. Soc. {\bf 68} (1950), 337--404.

\bibitem{BaG} A.O. Barut, L. Girardello: {\it New ``coherent'' states
associated with non-compact groups,\/} Comm. Math. Phys. {\bf21} (1971),
41--55.

\bibitem{BE} H. Bateman, A. Erd\'elyi, {\it Higher transcendental functions,
vol.~1--3,\/} McGraw-Hill, New York 1953--1955.

\bibitem{BeQ} F.A. Berezin: {\it Quantization,\/} Math. USSR Izvestiya
{\bf8} (1974), 1109--1163.

\bibitem{BMS} M. Bordemann, E. Meinrenken, M. Schlichenmaier:
{\it Toeplitz quantization of K\"ahler manifolds and $gl(n)$,
$n\to\infty$ limits,\/} Comm. Math. Phys. {\bf 165} (1994), 281--296.

\bibitem{E12} M. Engli\v s: {\it Berezin quantization and reproducing kernels
on complex domains,\/} Trans. Amer. Math. Soc. {\bf 348} (1996), 411--479.

\bibitem{E52} M. Engli\v s: {\it Berezin and Berezin-Toeplitz quantizations
for general function spaces,\/} Rev. Mat. Complut. {\bf 19} (2006), 385--430.

\bibitem{Fed} M.V. Fedoryuk: {\it Asymptotics, integrals, series\/}
(in~Russian), Nauka, Moscow, 1987.

\bibitem{ghan} A. Ghanmi: {\it A class of generalized complex Hermite
polynomials\/}, J. Math. Anal. {\bf 340} (2008), 1395--1406.

\bibitem{iszeng} M.E.H. Ismail, J. Zeng: {\it Two variable extensions of the
Laguerre and disc polynomials\/}, preprint (2014).

\bibitem{iszhang} M.E.H. Ismail, R. Zhang: {\it Classes of bivariate orthogonal
polynomials\/}, preprint (2014).

\bibitem{jotha} K. Jotsaroop, S. Thangavelu {\it Toeplitz operators with
special symbols on Segal-Bargmann spaces,\/} Integr. Equ. Oper. Theory
{\bf 69} (2011), 317–-346.

\bibitem{KarbSchl} A.V. Karabegov, M. Schlichenmaier: {\it Identification of
Berezin-Toeplitz deformation quantization,\/} J.~reine angew. Math. {\bf 540}
(2001), 49--76.

\bibitem{karp} D. Karp: {\it Square summability with geometric weights for
classical orthogonal expansions\/,} Advances in Analysis,
(H.G.W.~Begehr et~al, eds.), World Scientific, Singapore (2005), pp.~407--422.

\bibitem{odzhor} A. Odzijewicz, M. Horowski: {\it Positive kernels and
quantization,\/} J. Geom. Physics {\bf 63} (2013), 80-98. 

\bibitem{prugov} E. Prugove\v cki: {\it Consistent formulation of relativistic
dynamics for massive spin-zero particles in external fields,\/} Phys. Rev. D
{\bf 18} (1978), 3655--3673. 

\bibitem{Schwi} N. Schwid: {\it The~asymptotic forms of the Hermite and Weber
functions,\/} Trans. Amer. Math. Soc. {\bf 37} (1935), 339--362.

\bibitem{thanga} S. Thangavelu: {\it Hermite and Laguerre semigroups: some
recent developments\/,} Orthogonal families and semigroups in analysis and
probability, pp.~251-284, S\'emin. Congr.~25, Soc. Math. France, Paris, 2012. 

\bibitem{vanEijnd} S. J. L. van Eindhoven, J.L.H. Meyers:
{\it New orthogonality relations for the Hermite polynomials and related
Hilbert spaces,\/} J. Math. Anal. Appl. {\bf 146} (1990), 89--98.

\bibitem{HXua} H.~Xu: {\it An explicit formula for the Berezin star product,\/}
Lett. Math. Phys. {\bf 101} (2012), 239--264.

\end{thebibliography}
\end{document}